\documentclass[onecolumn,conference]{IEEEtran}

\usepackage{amsmath,amsthm,amssymb}
\usepackage{color}
\usepackage{url}
\usepackage{hyperref}
\usepackage{graphicx} 
\usepackage{mathrsfs}
\usepackage{epsfig}
\usepackage{verbatim}
\usepackage{setspace}

\newtheorem{theorem}{Theorem}
\newtheorem{lemma}[theorem]{Lemma}

\newtheorem{definition}[theorem]{Definition}
\newtheorem{claim}[theorem]{Claim}

\title{Correlated Jamming in a Joint Source Channel Communication System}
\author{
  \IEEEauthorblockN{Amitalok J. Budkuley and Bikash Kumar Dey }
  \IEEEauthorblockA{
    Indian Institute of Technology Bombay\\
		Mumbai, India\\
    Email: \{amitalok,bikash\}@ee.iitb.ac.in \vspace*{-0.5cm}}
  \and
  \IEEEauthorblockN{Vinod M. Prabhakaran}
  \IEEEauthorblockA{
     Tata Institute of Fundamental Research\\
     Mumbai, India\\
     Email: vinodmp@tifr.res.in \vspace*{-0.5cm}}
}

\begin{document}
\maketitle 
\begin{abstract}
We study correlated jamming in joint source-channel communication systems. An i.i.d. source is to be communicated over a memoryless channel in the presence of a correlated jammer with non-causal knowledge of user transmission. This user-jammer interaction is modeled as a zero sum game. A set of conditions on the source and the channel is provided for the existence of a Nash equilibrium for this game, where the user strategy is uncoded transmission and the jammer strategy is i.i.d jamming. 
This generalizes a well-known example of uncoded communication of a Gaussian sources over Gaussian channels with additive jamming. Another example, of a Binary Symmetric source over a Binary Symmetric channel with jamming, is provided as a validation of this result. 
\end{abstract}
\section{Introduction}\label{sec:introduction}
In this work, we examine the problem of communication in the presence of correlated jamming in Joint Source-Channel Communication (JSCC) systems. Here, an independent and identically distributed (i.i.d.) source is to be communicated over a memoryless channel while a jamming adversary, which can listen to the user communication, attempts to disrupt this communication. In the absence of jamming, we know that the separation scheme~\cite{shannon}, which is the strategy of splitting the coding into two stages of source coding followed by channel coding, results in optimal performance. The separation scheme generally needs the use of long codewords thus, causing large delays. However, Gastpar et al.\ show in~\cite{gastpar} that this is not always necessary, and that there exist certain JSCC systems where the optimal performance guaranteed through the separation principle can also be obtained through a much simpler zero-delay uncoded communication scheme. Such systems are called \textit{matched source-channel} systems in~\cite{gastpar}.

In one of the earliest works which studies jamming in JSCC systems~\cite{basar}, Ba\c{s}ar
considers the problem of communicating a zero-mean Gaussian random variable over an AWGN channel in the presence of a correlated jamming adversary. The jammer is average power constrained and allowed to correlate to the encoder's output i.e. the jammer has access to the encoder's input to the channel. The user-jammer interaction is modeled as a game over the average distortion in the source data. A Nash equilibrium pair of strategies for the user and the jammer is determined and the corresponding equilibrium utility value obtained. 
In our work, we extend this formulation to examine the problem of communicating an i.i.d. source vector over a memoryless channel under jamming. 

Apart from~\cite{basar}, various other works too, have considered the problem of jamming in JSCC systems. Ba\c{s}ar and Wu~\cite{basar-wu} study the problem of communicating a Gaussian random variable over an AWGN channel when the jammer, instead of being correlated to the encoder's output as in~\cite{basar}, is correlated to the encoder's input. The problem of jamming in additive Gaussian channels for arbitrary sources is studied in~\cite{ekyol-cdc2013},~\cite{ekyol-isit2013}. 

In this paper, we consider the problem of communicating a fixed length random vector produced by an i.i.d source over a  memoryless channel in the presence of a jammer. The jammer has non-causal knowledge of the user's encoder output (i.e., encoder's input to the channel) and is allowed to correlate with it. We formulate a non-cooperative zero sum game between the user and jammer, where the user's utility is the average source distortion. The user aims to minimize the average distortion while the jammer aims to maximize it. 
We determine a set of conditions on the source and the channel such that uncoded communication and i.i.d. jamming form a Nash equilibrium.
Extending the terminology introduced in~\cite{gastpar} to describe uncoded systems which are optimal, we call uncoded JSCC systems with i.i.d. jamming as \textit{matched source-jammer-channel} systems. We show that the Gaussian setup considered by Ba\c{s}ar in~\cite{basar} is an example of a matched source-jammer-channel system. We also study the problem of communicating a binary symmetric source over the binary symmetric channel in the presence of a correlated jammer and show it to be another instance of a matched source-jammer-channel system system. We determine its Nash equilibrium utility.

The following is the organization of the paper. In Section~\ref{sec:system:model}, we describe the system model and discuss the problem setup along with the resulting non-cooperative game. Our main result is stated in Section~\ref{sec:main:result} where we also provide its detailed proof through the analysis of the non-cooperative game between the user and the jammer. We discuss and analyse two important examples of JSCC systems with jamming, viz., the Gaussian system and the binary system, in Section~\ref{sec:examples}, the latter not having been considered before to the best of our knowledge. We devote the next two sections toward the discussion of certain important implications of our results and some overall concluding remarks. 
\section{System Model and Problem Statement}\label{sec:system:model}
\subsection{The Communication Setup}
Consider the communication setup given in Fig.~\ref{fig:general:setup}. 
\begin{figure}[!ht]
  \begin{center}
    \includegraphics[trim=0cm 10cm 0cm 0cm, scale=0.3]{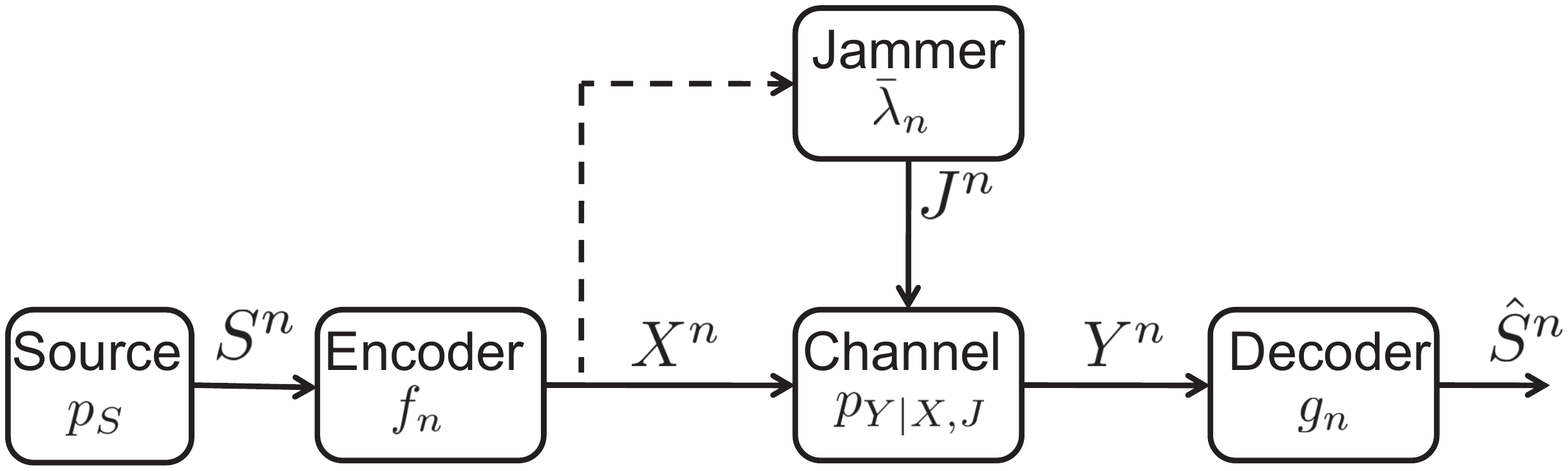}
    \caption{The Problem Setup}
    \label{fig:general:setup}
  \end{center}
\end{figure}
The user intends to transmit data from a Discrete Memoryless Source (DMS) over a Discrete Memoryless Channel (DMC) in the presence of a correlated jammer. The user encodes a size-$n$ block of source data $S^n$ into $X^n$ and transmits it over the channel. The correlated jammer observes this encoded block of data non-causally and inputs a jamming signal $J^n$ into the channel. Owing to the channel noise and the jammer input, the channel generates a noisy output $Y^n$ of the encoded data. The decoder, upon receiving $Y^n$, generates an estimate $\hat{S}^n$ of the original source data according to a given fidelity criterion. We now describe the important elements in this problem, viz., the source, the channel and the code. 

The DMS is specified by a pmf $p_S(s)$, $s \in \mathcal{S}$, where every $S_i$, $i=1,2,\dots,n$ is according to $p_S(.)$. The receiver estimate of the source data is $(\hat{S}_1,\hat{S}_2,\dots,\hat{S}_n)$, where $\hat{S}_i\in \hat{\mathcal{S}}$, for $i=1,2,\dots,n$. Consider an additive distortion criterion and let $d$ denote the per-letter distortion measure for the source data, where $d:\mathcal{S}\times\hat{\mathcal{S}}\rightarrow \mathbb{R}^+$. Here, the average distortion between source date $s^n$ and its estimate $\hat{s}^n$ is given as 
\begin{equation}\label{eq:d:definition}
d(s^n,\hat{s}^n)=\frac{1}{n} \sum_{i=1}^n d(s_i,\hat{s}_i).
\end{equation}
The communication channel is a DMC specified by the conditional pmf $p_{Y|X,J}(y|x,j)$, where $x\in \mathcal{X}$ is the encoder's channel input, $j\in \mathcal{J}$ is the  channel input of the jammer and $y\in \mathcal{Y}$ is the channel output respectively. Since the channel is a DMC, if $x^n$ and $j^n$ are the encoder's and the jammer's channel input respectively and the resulting channel output is $y^n$, then
\begin{equation*}\label{eq:DMC}
p_{Y^n|X^n,J^n}(y^n|x^n,j^n)=\prod_{i=1}^n p_{Y|X,J}(y_i|x_i,j_i).
\end{equation*}
The sender employs a \textit{unit-rate} joint source channel code~\cite{gallager} to communicate the source data over the channel. Such a joint source-channel code, denoted by $(f,g)$, is specified through the encoding function $f$, where $f:\mathcal{S}^n\rightarrow \mathcal{X}^n$, and a decoding function $g$, where $g:\mathcal{Y}^n\rightarrow \hat{\mathcal{S}}^n$. Let the cost of transmitting $x\in\mathcal{X}$ on the channel be $\rho_X(x)$, where $\rho_X:\mathcal{X}\rightarrow\mathbb{R}^+$. The user has an average transmit cost constraint of $P_U$ and hence, $X^n$ is such that $\sum_{i=1}^n \rho_X(X_i) \leq n P_U$. 

The encoder output is observed non-causally by an adversarial jammer and hence, the jammer can correlate its signal with it. The jammer behaviour is represented through its jamming function $\bar{\lambda}$.
Thus, upon observing $X^n$, the jammer chooses its own channel input $J^n=\bar{\lambda}(X^n)$, where $J^n=(J_i)_{i=1}^n=\bar{\lambda}(X^n)=(\lambda_1(X^n),\lambda_2(X^n),\dots,\lambda_n(X^n))$ and $J_i\in\mathcal{J}$, $i=1,2,\dots,n$. Let $\rho_{J|X}(j|x)$ be the jammer cost function, where $\rho_{J|X}: \mathcal{J}\times\mathcal{X}\rightarrow \mathbb{R}$. The jammer's signal $J^n$ is average transmit cost constrained to ${P}_J$, and thus, $\sum_{i=1}^n \mathbb{E}[\mathbb{E}[\rho_{J|X}(J_i,X_i)]] \leq n P_J$. Note here that, in general, the jammer cost function depends on $X$. 

Given these different elements, let us define a JSCC System with a Jammer (JSCCSJ) as follows.
\begin{definition}[JSCCSJ]
A JSCC System with a Jammer (JSCCSJ) is defined as a tuple $(p_S,p_{Y|X,J},\rho_X,\rho_{J|X},d)$, where $p_S$ is the i.i.d. source distribution, $p_{Y|X,J}$ is the memoryless channel distribution, $\rho_X$ is the user cost function, $\rho_{J|X}$ is the jammer cost function and $d$ is the distortion measure.
\end{definition}

Now for this JSCCSJ, upon using the code $(f,g)$ and when the jamming strategy is $\bar{\lambda}$, the resulting average distortion in the original source data $S^n$, given its estimate $\hat{S}^n$, is 
\begin{equation}\label{eq:average:distortion:definition}
D(f,\bar{\lambda},g)=\frac{1}{n} \mathbb{E} \left[ d(S^n,\hat{S}^n)\right].
\end{equation}
Note that we allow the user strategy as well as the jammer strategy to be  either deterministic or probabilistic. 
\subsection{The Problem Statement}
Given the average cost constraint $P_U$, the user seeks to minimize the average distortion in~\eqref{eq:average:distortion:definition} through the choice of an appropriate code $(f,g)\in \mathcal{C}^n({P}_U)$. Here, $\mathcal{C}^n({P}_U)$ is the set of all feasible $n$-length unit-rate joint source channel codes i.e. codes where the user average cost constraint ${P}_U$ is satisfied. On the contrary, the jammer seeks to maximize the average distortion in~\eqref{eq:average:distortion:definition} through the choice of an appropriate jamming strategy $\bar{\lambda}\in \mathcal{L}^n({P}_J)$. Here, $\mathcal{L}^n({P}_J)$ is the set of all feasible $n$-length jammer strategies i.e. jammer strategies which satisfy the jammer cost constraint ${P}_J$. 

Given the opposing interests of the user and the jammer over the average distortion, we formulate and analyse a non-cooperative zero sum game between the user and the jammer. Referring to the underlying system, we call this game the JSCCSJ $(p_S, p_{Y|X,J}, \rho_X, \rho_{J|X}, d)$ game. Here, our interest lies in determining a Nash equilibrium pair $(f^*,\bar{\lambda}^*,g^*)$ of user and jammer strategies for this game such that
\begin{align}\label{eq:optimality}
D(f^*,\bar{\lambda},g^*)\leq D(f^*,\bar{\lambda}^*,g^*)\leq D(f,\bar{\lambda}^*,g)
\end{align}
holds over all $(f,g)\in\mathcal{C}^n({P}_U)$ and $\bar{\lambda}\in\mathcal{L}^n({P}_J)$.  
Also, since all Nash equilibria of a zero sum game possess the same utility~\cite{osborne}, the game has a unique Nash equilibrium utility. The system will be characterized in terms of this equilibrium average distortion $D(f^*,\bar{\lambda}^*,g^*)$. 
\section{The Main Result}\label{sec:main:result}
The main result of our work is the determination of a set of conditions on 
the source and the channel in a JSCCSJ under which a 
single-letter coding strategy and i.i.d. jamming strategy form a Nash
equilibrium of the game over a JSCCSJ. The following theorem gives
the main result. Here $D(\cdot ||\cdot)$ is the Kullback-Leibler
divergence between two distributions.
\begin{theorem}\label{thm:main:result}
Let $p_S$ and $p_{Y|X,J}$ be a source and channel respectively.
For a given $p_{X|S}$, $p_{J|X}$ and $p_{\hat{S}|Y}$, suppose $\rho_X(x)$, $\rho_{J|X}(j|x)$ and  $d(s,\hat{s})$ are of the form
\begin{subequations}\label{eq:best:setup}
\begin{IEEEeqnarray}{rCl}
\rho_X(x) && \begin{cases}
    = a_1 D(p_{Y|X}(.|x)||p_{Y}(.))+a_2   & \mbox{if } p_{X}(x)>0 \\
    \geq a_1 D(p_{Y|X}(.|x)||p_{Y}(.))+a_2 & \mbox{ else }
  \end{cases} \label{eq:besta}\\	
d(s,\hat{s})&&\hspace{8pt}=-b_1\log (p_{S|\hat{S}}(s|\hat{s}))+d_0(s)\label{eq:bestb}\\
\rho_{J|X}(j|x)&& \begin{cases}
    =c_1 \sum_{s,y,\hat{s}} p_S(s)p_{X|S}(x|s) p_{Y|X,J}(y|x,j) p_{\hat{S}|Y}(\hat{s}|y)d(s,\hat{s})+c_2  & \mbox{ if } p_{J|X}(j|x)>0 \\
    \geq c_1 \sum_{s,y,\hat{s}} p_S(s)p_{X|S}(x|s) p_{Y|,X,J}(y|x,j) p_{\hat{S}|Y}(\hat{s}|y)d(s,\hat{s})+c_2 & \mbox{ else }
  \end{cases} \label{eq:bestc}
\end{IEEEeqnarray}
\end{subequations}
%
%
    %
%
for some constants $a_1> 0$, $b_1> 0$, $c_1>0$, $a_2$, $c_2$, and some function $d_0(s)$. 
Then, the single-letter communication strategy
$(p_{X|S}, p_{\hat{S}|Y})$ and the i.i.d. jamming strategy $p_{J|X}$ form a Nash equilibrium for the JSCCSJ $(p_S, p_{Y|X,J}, \rho_X, \rho_{J|X}, d)$ game.
Here, $p_X(x)=\sum_{s} p_{S}(s)p_{X|S}(x|s)$.
\end{theorem}
Such a single-letter coding strategy is commonly referred to as an 
\textit{uncoded communication} strategy. Extending the nomenclature 
in~\cite{gastpar}, we call a JSCCSJ a \textit{matched source-jammer-channel} 
system if uncoded communication strategy and i.i.d. jamming form
a Nash equilibrium.

{\em Proof of Theorem~\ref{thm:main:result}:}
The proof has two parts.\\
{\em Part A:} In this part, we show that when the jamming signal is i.i.d. with a pmf $p_{J|X}$, if $\rho_X$ and $d$ are chosen as in~\eqref{eq:best:setup}, then the single-letter  
communication scheme $(p_{X|S}, p_{\hat{S}|Y})$ minimizes the average 
distortion.

Since the channel $p_{Y|X,J}$ is a memoryless channel and the jammer strategy $p_{J|X}$ is an i.i.d strategy, the resulting channel between the encoder and the decoder is also a memoryless channel with
\begin{equation*}\label{eq:Y:X:when:jammer:iid}
p_{Y|X}(y|x)=\sum_{j\in \mathcal{J}}p_{Y|X,J}(y|x,j)p_{J|X}(j|x).
\end{equation*}
For this channel transition probability, the conditions~\eqref{eq:besta} and~\eqref{eq:bestb} are the same conditions
as in~\cite[Theorem 6]{gastpar} under which the uncoded communication strategy
$(p_{X|S}, p_{\hat{S}|Y})$ minimizes the average distortion. Thus,
the result follows from~\cite[Theorem 6]{gastpar}.

\noindent
{\em Part B:}
In this part, we show that when the user employs uncoded communication 
strategy $(p_{X|S},p_{\hat{S}|Y})$, and $d$ and $\rho_{J|X}$ are
as in~\eqref{eq:bestb} and~\eqref{eq:bestc}, then i.i.d jamming strategy 
$p_{J|X}$ maximizes the average distortion.

Let us fix the user strategy to be uncoded communication where the encoder and decoder pmfs are $p_{X|S}$ and $p_{\hat{S}|Y}$ respectively.

To begin with, for a given distortion measure $d$, jammer's
cost function $\rho_{J|X}$, and a fixed uncoded user strategy
$(p_{X|S},p_{\hat{S}|Y})$, 
let us define the jammer's \textit{distortion cost function} as 
\begin{equation}\label{eq:jammer:distortion:cost:function}
D_{p_{X|S},p_{\hat{S}|Y}}({P}_J)=\max_{p_{J|X}:\mathbb{E}[\mathbb{E}[\rho_{J|X}(J|X)]]\leq {P}_J} \mathbb{E}[d(S,\hat{S})].
\end{equation} 
For simplicity of notation, we denote this distortion cost function
as $D({P}_J)$ without the subscript. Note that the above maximization
is over i.i.d. jamming strategies, and it is not obvious that the jammer
can not exceed this distortion by employing non-i.i.d. jamming. 
The following theorem establishes this non-trivial result.
\begin{theorem}\label{thm:jammer:iid}
For a given source $p_S$, a channel $p_{Y|X,J}$, an uncoded user strategy $(p_{X|S},p_{\hat{S}|Y})$, a distortion measure $d$, and a jammer cost function
$\rho_{J|X}$; the maximum average distortion (per source data bit) inflicted by the correlated jammer is $D({P}_J)$.
\end{theorem}	
Before proving this result, we state and prove some auxiliary results on $D(P_J)$.
\begin{claim}
$D(P_J)$ is a non-decreasing function.
\end{claim}
\begin{proof}
The proof follows in a straightforward manner.
\end{proof}
\begin{claim}
$D(P_J)$ is a linear function.
\end{claim}
\begin{proof}
Let  $(P_{J,1},D_1)$ and $(P_{J,2},D_2)$ be such that 
\begin{equation}\label{eq:D:i:P:i}
D_i= \max_{p_{J|X}:\mathbb{E}[\mathbb{E}[\rho_{J}(J)]]\leq P_{J,i}} \mathbb{E}[d(S,\hat{S})]
\end{equation}
and let $p_{J|X}^1$ and $p_{J|X}^2$ be the corresponding maximizing distributions in~\eqref{eq:D:i:P:i}. Now let us define a new jamming distribution $p_{J|X}^\beta$ where, $p_{J|X}^\beta=\beta p_{J|X}^1+(1-\beta)p_{J|X}^2$. Also, let $P_{J,\beta}=\beta P_{J,1}+(1-\beta) P_{J,2}$. Since $\mathbb{E}[d(S,\hat{S})]$ is a linear function of $p_{J|X}$, therefore $D(P_{J,\beta})=\beta D(P_{J,1})+(1-\beta) D(P_{J,2})$. Hence, it follows that $D(P_J)$ is linear in $P_{J}$.
\end{proof}
We now give the proof of Theorem~\ref{thm:jammer:iid}.
\begin{proof}[Proof of Theorem~\ref{thm:jammer:iid}]
To prove this result, it is sufficient to establish that the following lemma is true. 
\begin{lemma}
For a given source $p_S$, a channel $p_{Y|X,J}$, an uncoded user strategy $(p_{X|S},p_{\hat{S}|Y})$, a distortion measure $d$, a jammer cost function
$\rho_{J|X}$; and any jamming strategy $\bar{\lambda}\in\mathcal{L}^n(P_J)$ which results in average (per-letter) distortion $D$, $D\leq D(P_J)$.
%
\end{lemma}
\begin{proof}
For the average per-letter distortion $D$, we have
\begin{subequations}
\begin{eqnarray}
                                   D&\stackrel{(a)}= & \frac{1}{n}\mathbb{E}[d(S^n,\hat{S}^n)]\\
                                    &\stackrel{(b)}=& \frac{1}{n}\sum_{i=1}^n  \mathbb{E}[d(S_i,\hat{S}_i)]\\
                                    &\stackrel{(c)}\leq & \frac{1}{n}\sum_{i=1}^n  D\left(\mathbb{E}[\mathbb{E}[\rho_{J|X}(J_i|X_i)]]\right)\\ 
																		&\stackrel{(d)}= &   D\left(\frac{1}{n}\sum_{i=1}^n \mathbb{E}[\mathbb{E}[\rho_{J|X}(J_i|X_i)]]\right)\\ 
															&\stackrel{(e)}\leq &D(P_J)
\end{eqnarray}
\end{subequations}
Here, $(a)$ and $(b)$ follow from~\eqref{eq:average:distortion:definition} and ~\eqref{eq:d:definition} respectively. For $(c)$, we use the definition of $D(P)$ in~\eqref{eq:jammer:distortion:cost:function}. Finally, while $(d)$ follows from the linearity of $D(P)$ and the linearity of expectation, $(e)$ is true owing to the non-decreasing nature of $D(P_J)$. 
\end{proof}
This shows that i.i.d jamming maximizes average distortion and hence, concludes the proof of Theorem~\ref{thm:jammer:iid}.
\end{proof}
We will now show in the following lemma, that under the setup of Theorem~\ref{thm:jammer:iid}, 
if $d$ and $\rho_{J|X}$ are of the form~\eqref{eq:bestb} and~\eqref{eq:bestc}
respectively, then the i.i.d. jamming strategy $p_{J|X}$ maximizes
the average distortion.
\begin{lemma}\label{lem:jammer:cost:function}
For a given source $p_S$, a channel $p_{Y|X,J}$, an uncoded user strategy $(p_{X|S},p_{\hat{S}|Y})$, if $d$ and $\rho_{J|X}$ are as in~\eqref{eq:bestb}
and~\eqref{eq:bestc}; then the i.i.d. jammer strategy $p_{J|X}$ 
achieves $D({P}_J)$, where ${P}_J=\mathbb{E}[\rho_{J|X}(J|X)]$ with
the expectation taken over $p_Sp_{X|S}p_{J|X}$.
\end{lemma}
\textit{Proof:}
Let $\tilde{p}_{J|X}$ be another feasible i.i.d. jamming  strategy.
Let $\Delta_{p_{J|X}}$ and $\Delta_{\tilde{p}_{J|X}}$ denote the average distortion
resulting from the jammer strategies $p_{J|X}$ and $\tilde{p}_{J|X}$
respectively. Then, we have
\begin{align*}
& \Delta_{p_{J|X}} - \Delta_{\tilde{p}_{J|X}}\nonumber\\
&\,\,=\sum_{x,j}\left(\sum_{s,y,\hat{s}} p_S(s)p_{X|S}(x|s) p_{Y|X,J}(y|x,j) \right. \nonumber\\
&\hspace*{15mm}\left. \rule{0pt}{7mm}p_{\hat{S}|Y}(\hat{s}|y)d(s,\hat{s})\right) (p_{J|X}(j|x) - \tilde{p}_{J|X}(j|x))
\end{align*}
Now, let us define $h(x,j) = (\rho_{J|X}(j|x) - c_2)/c_1$. It is clear
from the expression of $\rho_{J|X}$ in~\eqref{eq:bestc} that
\begin{align*}
 \Delta_{p_{J|X}} - \Delta_{\tilde{p}_{J|X}}
&\geq \sum_{x,j} h(x,j) (p_{J|X}(x|j) - \tilde{p}_{J|X}(x|j))\nonumber \\
& = \frac{1}{c_1} \sum_{x,j} \rho_{J|X}(j|x) (p_{J|X}(x|j) - \tilde{p}_{J|X}(x|j))\nonumber
\end{align*}
The first term is ${P}_J$ by definition, and the second term is
the expected cost incurred by the i.i.d. jamming strategy $\tilde{p}_{J|X}$.
Since $\tilde{p}_{J|X}$ is also a feasible i.i.d. jamming strategy
with expected cost at most ${P}_J$, we have
\begin{equation*}
 \Delta_{p_{J|X}}-\Delta_{\tilde{p}_{J|X}}  \geq 0. ~\QEDopen
\end{equation*}
For a given user strategy, as assumed in Lemma~\ref{lem:jammer:cost:function}, there may be
multiple feasible i.i.d. jamming strategies which are optimal. In fact,
it can be noted from the above proof that if the definition of
$\rho_{J|X}$ in~\eqref{eq:bestc} holds with equality for all $(x,j)$
with $p_{J|X}(j|x)=0$, then any i.i.d. jamming strategy $\tilde{p}_{J|X}$
which utilizes the full cost ${P}_J$ is optimal. However, for such
$\tilde{p}_{J|X}$, the user strategy $(p_{X|S},p_{\hat{S}|Y})$ may not be
optimal, and thus, this pair of strategies may not form a Nash equilibrium.
\section{Some Important Illustrations of Theorem~\ref{thm:main:result}}\label{sec:examples}
We now provide some example JSCCSJ setups that  are 
\textit{matched source-jammer-channel} systems. We first consider a Gaussian JSCCSJ where a Gaussian source is communicated over an AWGN channel under
the squared error distortion measure, and quadratic cost function 
for the encoder and the jammer.  Next, we consider a Binary JSCCSJ where a binary 
symmetric source is sent over a binary symmetric channel under the
Hamming distortion measure and constant cost function for the encoder and fixed cost function for the
jammer.
\subsection{The Gaussian JSCCSJ}
In the Gaussian JSCCSJ, the source is i.i.d. Gaussian with $p_S\sim \mathcal{N}(0,1)$ (can be easily generalized to arbitrary variance). The channel is an AWGN channel with $p_{Y|X,J}\sim\mathcal{N}(X+J,\sigma^2)$.  
In addition, the source distortion measure $d(s,\hat{s})=(s-\hat{s}^2)$,
the encoder cost function $\rho_X(x)=x^2$, and the jammer cost function 
$\rho_{J|X}(j|x)=j^2$.  Let us consider an uncoded communication strategy
with deterministic encoder $f(S)=\sqrt{P_U}S$ and decoder 
$g(Y)=P_U/(P_U+\sigma^2)Y$. 
They correspond to 
$p_{X|S}(x|s) = I_{\{x=\sqrt{P_U}s\}}$ and 
$p_{\hat{S}|y}(\hat{s}|y)=I_{\{\hat{s}=y P_U/(P_U+\sigma^2)\}}$,
where $I$ denotes the indicator function. Let us also consider the 
i.i.d. jamming strategy $J_i = \alpha X_i +R_i$, $\alpha$ is a constant and $R_i$, $i=1,2,\dots,n$ are i.i.d. with each $R_i\sim\mathcal{N}(0,\sigma_R^2)$.
So, conditioned on $X$, the jammer signal $J$ has the distribution
$\mathcal{N}(\alpha X,\sigma_R^2)$. Here $\alpha$ and $\sigma_R^2$ 
depend on $\sigma^2$, $P_U$ and $P_J$
as given in~\cite[Theorem 1]{basar}. It can be checked that for these
user and jammer strategies, $d, \rho_X,$ and $\rho_{J|X}$ satisfy
\eqref{eq:best:setup}. Hence we have
\begin{lemma}\label{lem:gaussian:setup}
The Gaussian JSCCSJ is a matched system, and the uncoded communication with encoder $f(S)=\sqrt{P_U}S$ and decoder $g(Y)=P_U/(P_U+\sigma^2)Y$ and i.i.d. linear Gaussian jamming
form a Nash equilibrium.
\end{lemma}

Note that~\cite[Theorem 1]{basar} is a special case of the Gaussian JSCCSJ considered in Lemma~\ref{lem:gaussian:setup} for $n=1$. 
\subsection{The Binary JSCCSJ}
In the binary JSCCSJ, an i.i.d binary symmetric source with 
$p_S\sim Bern(1/2)$, is to be transmitted over a binary symmetric 
channel $Y=X\oplus J\oplus Z$, where $Z\sim Bern(p)$ is an independent noise with
$p\leq 1/2$ (without loss of generality; otherwise things change slightly
as pointed out later), $J \in \{0,1\}$, and $\oplus$ denotes modulo-$2$ addition.
In addition, the source has Hamming distortion measure 
$d(s,\hat{s})=d_{H}(s,\hat{s}) =I_{s\neq\hat{s}}$,
the encoder has constant cost function $\rho_X(x)=k$, and the binary jammer has 
cost function $\rho_{J|X}(j|x)=j$.  Let the jammer have a cost constraint
$P_J\leq 1/2$. This restricts the jammer
to input at most $P_J$ fraction of $1$s. Note that $p*P_J :=
p(1-P_J)+P_J(1-p) \leq 1/2$
so that the jammer and the independent noise together can not flip more
than half of the transmitted bits.
Let us consider an uncoded communication strategy
with deterministic encoder $f(S)=S$ and decoder 
$g(Y)=Y$. They correspond to 
$p_{X|S}(x|s) = I_{\{x=s\}}$ and 
$p_{\hat{S}|y}(\hat{s}|Y)=I_{\{\hat{s}=y\}}$.
Let us also consider the 
i.i.d. jamming strategy (independent of encoder output on the channel) with
$J_i \sim Bern(P_J)$. 

It can be checked that for these
user and jammer strategies, $d, \rho_X,$ and $\rho_{J|X}$ satisfy~\eqref{eq:best:setup}. Hence we have the following lemma.
\begin{lemma}\label{lem:binary:setup}
The Binary JSCCSJ is a matched system, and the uncoded communication
with encoder $f(S)=S$ and decoder 
$g(Y)=Y$ and i.i.d. Bernoulli jamming independent of state
form a Nash equilibrium.
\end{lemma}
Since the binary JSCCSJ satisfies the conditions in Theorem~\ref{thm:main:result}
for uncoded user communication and i.i.d. jamming
independent of $X$ as mentioned above, it proves the lemma.
This shows that the binary JSCCSJ is implicitly connected
to the above user and jammer strategies.

In the following, we present a more specific elementary proof
of Lemma~\ref{lem:binary:setup}. We proceed by validating~\eqref{eq:optimality} for this pair of strategies. 
We first establish the RHS and then, proceed to establish the LHS of~\eqref{eq:optimality}.

{\em Proof of Lemma~\ref{lem:binary:setup}:}\\
{\em The RHS of inequality~\eqref{eq:optimality}:}
To prove the RHS of~\eqref{eq:optimality}, let us fix the jammer strategy to be i.i.d.  $Bern(P_J)$, 
We now have an equivalent BSC with a transition 
probability $\hat{p}=p(1-P_J)+(1-p) P_J$. 
Thus, we now have the problem of communicating a binary symmetric source $S$ over a BSC($\hat{p}$) with $\hat{p} \leq 1/2$.
From~\cite{gastpar},~\cite{gallager}, we know that if $\rho_X(x)$ is a constant
function, and the source distortion measure is the Hamming distortion measure, the single-letter coding scheme, where $f(S)=S$ and $g(Y)=Y$, is an optimal communication scheme which minimizes the average distortion. This validates the RHS of~\eqref{eq:optimality} for the binary JSCCSJ.

{\em The LHS of inequality~\eqref{eq:optimality}:}
Here, we assume that the user strategy is the single-letter coding strategy 
where $f(S)=S$ and $g(Y)=Y$.  From Theorem~\ref{thm:jammer:iid}, we 
know that for this user strategy of uncoded communication, the best 
jamming strategy is an i.i.d. jamming strategy. 
The best jamming strategy is the solution of the 
optimization problem in~\eqref{eq:jammer:distortion:cost:function} under
the cost constraint $P_J$. We 
now show that for the given jammer cost function, the best jamming 
strategy is a Bern($P_J$) jamming strategy independent of $X$. 
Since $\rho_{J|X}(j|x) = j$, the best jamming strategy is the solution of the following problem. 
\begin{align*}\label{eq:jammer:opt:problem:2}
& \arg \max_{p_{J|X}:\mathbb{E}[J]\leq P_J} \mathbb{E}[d_H(S,\hat{S})]\\
= & \arg \max_{p_{J|X}:\mathbb{E}[J]\leq P_J} \mathbb{E}[d_H(S,S\oplus J\oplus Z)]\\
= & \arg \max_{p_{J}:\mathbb{E}[J]\leq P_J} \mathbb{E}[J\oplus Z]\\
= & \arg \max_{p_{J}:\mathbb{E}[J]\leq P_J} Pr\{J\oplus Z=1\}
\end{align*}
This follows by noting that $\hat{S}=S\oplus J \oplus Z$ for the
given encoder-decoder pair,
where $Z\sim Bern(p)$. This shows that $J$ can be chosen independent of
$X$.
Clearly, under the given jammer cost constraint, the above maximum
is achieved by $J \sim Bern(P_J)$.

This completes the proof.
\section{Discussion}
We see that the matching in a JSCCSJ depends on the following: the source (with distortion measure), the channel (with user and jammer cost functions), the single-letter user coding scheme along with the i.i.d. jammer strategy. In~\cite{gastpar}, the results in Lemmas 3 and 4 of that paper (not reproduced here) were instrumental in determining the conditions on the encoder cost function and distortion measure for matching in JSCC systems with no jammer. These ``inverse'' lemmas show that there is a choice of the cost function and cost constraint (distortion measure and distortion constraint, resp.) under which a given input distribution (test channel distribution, resp.) is optimal for the channel (source, resp.) coding problem. In establishing the conditions for matching in JSCCSJ in our work, we see that, along with \cite[Lemmas 3 and 4]{gastpar}, our Lemma~\ref{lem:jammer:cost:function} also plays an important role. Lemma~\ref{lem:jammer:cost:function} is also an ``inverse'' lemma in that it establishes that there is a choice for the jammer cost function ($\rho_{J|X}$) and cost constraint ($P_J$) such that a given i.i.d. jamming strategy is optimal against a (given) uncoded user strategy. 

In our example, the equilibrium strategies for the user and the jammer are linear strategies. However, this does not rule out the existence of other pairs of equilibrium strategies, possibly non-linear. This, and other related questions, are investigated in~\cite{ekyol-cdc2013},~\cite{ekyol-isit2013}. 

Finally, the binary JSCCSJ extends straight forwardly to the $L-$ary uniform system to give more examples of discrete alphabet matched JSCCSJs.
\section{Conclusion}
In this work, we studied the effects of correlated jamming on the performance of a JSCC system. The user-jammer interaction was modeled as a zero sum game over the average distortion and the optimal performance of the system was characterized through the Nash equilibrium utility of the game. A set of conditions on the source and the channel were determined for the existence of an equilibrium, where the equilibrium user strategy was uncoded communication and the equilibrium jammer strategy was i.i.d jamming. Thus, a probabilistic matching of the source, jammer and channel was shown to exist and such systems were called matched source-jammer-channel systems. It was shown that a well-known example of a Gaussian JSCC with jamming was a special case of our problem. In addition, another example of  Binary JSCC system with a jammer, was analysed and its solution provided. 
\section*{Acknowledgment}
The authors thank Sibiraj B. Pillai for fruitful discussions.
The work  was supported in part by the
Bharti Centre for Communication, IIT Bombay, a grant from the Information Technology Research Academy, Media Lab Asia, to IIT Bombay and TIFR, a grant from the Department
of Science and Technology, Government of India, to IIT Bombay, and the Ramanujan Fellowship from the Department of Science and Technology, Government of India, to V. M. Prabhakaran.
\addcontentsline{toc}{section}{Acknowledgment} 
\bibliographystyle{IEEEtran}
\bibliography{IEEEabrv,References}
\end{document}